\title{Commitment, Conflict, and Status Quo in Bargaining}
\author{Harry PEI\footnote{Department of Economics, Northwestern University. I thank Nina Bobkova, Stephen Morris, and Bruno Strulovici for helpful comments. I thank NSF Grant SES-2337566 for financial support.}}
\date{\today}
\documentclass[12pt]{article}

\usepackage{amsmath}
\usepackage{graphicx}
\usepackage{amsfonts}
\usepackage{amsthm}
\usepackage{setspace}
\usepackage{tikz}
\usepackage{amssymb}
\usetikzlibrary{patterns}
\usepackage{sgame}
\usepackage{color}
\usepackage{hyperref}
\usepackage{times}
\usepackage{enumitem}
\usepackage{csquotes}
\usepackage{multirow,array}
\usepackage[top=1in, bottom=1in, left=1in, right=1in]{geometry}
\begin{document}
\numberwithin{equation}{section}

\maketitle
\noindent \textbf{Abstract:} Each period, two players bargain over a unit of surplus. Each player chooses between remaining flexible and committing to a take-it-or-leave-it offer at a cost. If players' committed demands are incompatible, then the current-period surplus is destroyed in the conflict. When both players are flexible, the surplus is split according to the \textit{status quo}, which is the division in the last period where there was no conflict. We show that when players are patient and the cost of commitment is small, there exist a class of symmetric Markov Perfect equilibria that are asymptotically efficient and renegotiation proof, in which players commit to fair demands in almost all periods.\\


\newtheorem{Proposition}{\hskip\parindent\bf{Proposition}}
\newtheorem{Theorem}{\hskip\parindent\bf{Theorem}}
\newtheorem{Lemma}{\hskip\parindent\bf{Lemma}}
\newtheorem{Corollary}{\hskip\parindent\bf{Corollary}}
\newtheorem*{Definition}{\hskip\parindent\bf{Definition}}
\newtheorem*{Assumption}{\hskip\parindent\bf{Regular Records}}
\newtheorem*{Condition}{\hskip\parindent\bf{Condition}}
\newtheorem{Claim}{\hskip\parindent\bf{Claim}}
\newtheorem*{Assumption1}{\hskip\parindent\bf{Assumption 1'}}

\begin{spacing}{1.5}

\section{Introduction}\label{sec1}
In many bargaining situations, ranging from legislative bargaining to bargaining between business partners,  the allocation of surplus is sometimes shaped by past agreements. For example, when a party makes a compromise and accepts an unfair deal, it can establish a precedent for others to exploit in the future, pressuring them into accepting further unfavorable agreements. Aware of these consequences, forward-looking players  might be unwilling to make concessions at the bargaining table, even at the risk of bargaining breakdowns and lowering the overall surplus.

We analyze an infinite-horizon bargaining game where (i) players can make costly commitments regarding the share of the surplus they demand 
and (ii) past agreements between players will become the status quo for future negotiations. In our model, players can benefit from altering the status quo in their favor and such benefits can be arbitrarily large as players become arbitrarily patient. Nevertheless, we show that as the cost of commitment and the discounting vanish to zero, there exist equilibria where both players commit to receiving a half of the surplus in almost all periods. Our equilibria are renegotiation proof and Markov perfect, which rule out punishments that are Pareto inefficient as well as those that are contingent on payoff-irrelevant information.

In every period of our model, players bargain over a unit of divisible good by choosing whether to \textit{remain flexible} or to \textit{commit to some share} at a small cost. In practice, such commitments can take the form of burning bridges or making public announcements, which makes it very costly (if not impossible) for the parties who make those commitments to back down.\footnote{For example, if a union leader has repeatedly announced that they will step down unless they can increase the workers' salaries by $50\%$, then it would be very hard for them to accept a wage increase that is much less than $50\%$.} 
Each player's utility from consuming the good is a strictly increasing and strictly concave function of the share they receive, which means that in the efficient allocation, each player receives half a unit. 
If one player is committed and the other is flexible or committed to something compatible (e.g., both players commit to $1/2$), then the committed player receives the share they committed to and the rest belongs to the other player. Both players receive nothing in periods where they make incompatible commitments. If both of them are flexible, then the surplus is split according to the \textit{status quo}. The status quo in the initial period is exogenously given, but then it evolves endogenously over time: It coincides with the division of surplus in the last period where the surplus was split.

As a benchmark, in every Nash equilibrium of the one-shot bargaining game and in every Markov Perfect equilibrium of the repeated bargaining game where the status quo evolves exogenously, players will never commit to anything other than the entire surplus. This is because (i) committing to anything weakly less than the status quo is strictly dominated by being flexible, and (ii) after a player knows that their opponent will never commit to anything less than their share under the status quo, they will have no incentive to commit to anything strictly less than $1$ since it will be strictly dominated by some mixture between committing to $1$ and remaining flexible. As a result, players' payoffs are bounded below first best since the good is almost never equally shared.

In the repeated bargaining game with an endogenously evolving status quo, there exist Markov Perfect equilibria in which sufficiently patient players' equilibrium payoffs are close to first best as their cost of commitment vanishes. This conclusion may not be surprising since players can deter each other from demanding the entire surplus by coordinating on an inefficient continuation equilibrium once a player's share unless the status quo division is fair.\footnote{Richter (2014) uses a different dynamic model to study dynamic legislative bargaining, in which he constructs an equilibrium that achieves fair division in the long run. In the equilibrium he constructed, players will coordinate on Pareto inefficient punishments when the status quo does not give all players an equal share.} Such inefficient continuation equilibria exist. One example is that once a player's status quo share equals $1$, both players mix between committing to $1$ and remaining flexible, with the probabilities of committing to $1$ converging to $1$ as players are patient and their cost of commitment vanishes. If that is the case, then each player's continuation value is close to $0$ when their status quo allocation reaches $1$, which discourages them from committing to the entire surplus when the status quo allocation is fair.

Our main result shows that there exist a class of Markov Perfect equilibria in which as players become arbitrarily patient and their cost of commitment vanishes, (i) both players will commit to demand their fair share in almost all periods, which implies that their equilibrium payoffs are close to the socially optimum, and (ii) at every history (including those that are never reached on the equilibrium path), their stage-game payoffs and their continuation values are close to the Pareto frontier. This implies that the socially efficient allocation can be achieved in MPEs that are \textit{asymptotically renegotiation proof}, in the sense that at every history, players' gains from renegotiations are negligible in terms of their stage-game payoffs and continuation values. Our equilibria also have the realistic feature that once a player's status quo allocation is close to the entire surplus, then his continuation value is also close to that from receiving the entire surplus in all periods.

The key to our construction is how to provide incentives for players to commit to demand their fair share when the current status quo allocation is fair. If both players commit to their fair share, then one player will have an incentive to deviate by remaining flexible, which saves their commitment cost. When one player is flexible, the other player will have an incentive to demand the entire surplus since it will lead to a high continuation value. Hence, in order to provide the right incentives, both players need to randomize between demanding the entire surplus, demanding their fair share, and remaining flexible. Interestingly, the equilibrium probability with which each player demanding their fair share converges to one as they become arbitrarily patient and their cost of commitment vanishes.

\paragraph{Related Literature:} Our paper contributes to the literature on repeated bargaining, in which two players bargain over some surplus in every period and the game does not end after an offer is accepted. Compared to earlier works that focus on the role of incomplete information (e.g., Schmidt 1993, Lee and Liu 2013) or imperfect observability (e.g., Wolitzky 2023), we abstract away from these concerns and instead focus on the roles of costly commitments and the endogenously evolving status quo, which  to the best of our knowledge, is novel in this literature.

In terms of modeling commitment, we take Schelling (1960)'s perspective that views bargaining as a struggle between players to convince others that they are committed. Schelling's idea was first formalized by Crawford (1982), followed by Ellingsen and Miettinen (2008), Dutta (2024), and  Miettinen and Vanberg (2025).\footnote{These works differ from the reputational bargaining model of Abreu and Gul (2000) where each player has private information regarding whether they are committed and can build a reputation for being committed.} Compared to those works that focus on static settings, we examine a dynamic setting where the current division of surplus affects the future status quo. The  endogenously evolving status quo  and the possibility of conflict motivate players to make moderate demands, which leads to equilibria that are both socially efficient and renegotiation proof.

Despite many modeling differences (such as players making offers simultaneously as well as the possibility of conflicts and surplus destruction), one feature of our model that the current-period surplus division determines the status quo in future periods is reminiscent of the literature on dynamic legislative bargaining pioneered by Epple and Riordan (1987) and Baron (1996). Epple and Riordan (1987), Kalandrakis (2004), Bowen and Zahran (2012), Duggan and Kalandrakis (2012), Richter (2014), and Anesi and Seidmann (2015) study models where in each period, a group of players decide how to divide a unit of surplus, a player is randomly selected to make a proposal, and the previous-period division (i.e., the status quo) will be enacted once the proposal is rejected. In contrast, the entire surplus is destroyed once players make conflicting commitments and the status quo is enacted only when both players are flexible. 
Moreover, our efficiency result focuses on  renegotiation proof equilibria, which rule out Pareto inefficient punishments at off-path information sets. In contrast, the efficient equilibria constructed in Epple and Riordan (1987) and Richter (2014) hinge on Pareto inefficient punishments.\footnote{Epple and Riordan (1987) use the optimal penal code in Abreu (1988), which requires inefficient punishments at off-path histories. Richter (2014) constructs an MPE where fair division is achieved asymptotically in which some of the surplus is wasted at certain status quo. Other papers in this literature show that equal division cannot be achieved 
when players are sufficiently patient due to their incentive constraints, such as Bowen and Zahran (2012). 
Kalandrakis (2004) constructs equilibria in which each player demands the entire surplus when they can make offers.}

\section{Model}\label{sec2}
Time is discrete, indexed by $t=0,1,2,...$. There are two players $i \in \{1,2\}$ who discount future payoffs by the same factor $\delta \in (0,1)$. Each period, they bargain over a unit of divisible and perishable good
by choosing whether to remain \textit{flexible} (denote this action by $f$) or to \textit{commit} to a demand $s \in [1/2,1]$ at cost $c>0$. We use $s_{i,t}$ to denote player $i$'s demand in period $t$ if they commit to something. Hence, each player chooses his action from the set $S \equiv \{f\} \bigcup [1/2,1]$. 

Let $u: [0,1] \rightarrow [0,1]$ denote the mapping from a player's share of the good to their utility,
which is strictly increasing and strictly concave with $u(0)=0$ and $u(1)=1$.\footnote{Our main result, Theorem \ref{Theorem1}, holds for all strictly increasing and weakly concave $u$. Nevertheless, when $u$ is linear, asymptotic efficiency can also be attained in the one-shot game.}
Let $\mathcal{U} \equiv \{(u(\alpha),u(1-\alpha))\}_{\alpha \in [0,1]}$ denote the Pareto frontier.
The allocation that maximizes utilitarian welfare is one where each player receives $1/2$ unit, from which the sum of their payoffs is 
$2u(1/2)$.

Players' stage-game payoffs in period $t \in \mathbb{N}$ depend both on their actions and on that period's \textit{status quo}, denoted by $(\alpha_{1,t},\alpha_{2,t}) \in [0,1]^2$ with $\alpha_{1,t}+\alpha_{2,t}=1$. If both players are committed and their commitments are incompatible, i.e., $s_{1,t}+s_{2,t}>1$, then the good is destroyed in the conflict and both players receive stage-game payoff $-c$. If both players are committed and their commitments are compatible, i.e., $s_{1,t}=s_{2,t}=1/2$, then each player receives $1/2$ unit of the good and each player's payoff is $u(1/2)-c$. If one player $i$ commits to $s_{i,t}$ and the other is flexible, then the committed player's stage-game payoff is $u(s_{i,t})-c$ and the flexible player's stage-game payoff is $u(1-s_{i,t})$. If both players are flexible, then player $i \in \{1,2\}$'s stage-game payoff is $u(\alpha_{i,t})$.

Our key modeling innovation is that the current division of surplus affects the future outcomes by affecting the status quo in the next period, which we interpret as setting a \textit{precedent}.

The status quo in period $0$ is $(\alpha_0,1-\alpha_0)$ for some exogenously given $\alpha_0 \in (0,1)$. 
For every $t \geq 1$, the status quo in period $t$, denoted by $(\alpha_{1,t},\alpha_{2,t})$,
(i) coincides with the status quo in period $t-1$ when the good was not successfully divided in period $t-1$, that is, when both players were committed in period $t-1$ and chose $s_{1,t-1}+s_{2,t-1}>1$, and (ii) coincides with the division in period $t-1$ when the good was successfully divided in period $t-1$. This suggests that
$(\alpha_{1,t},\alpha_{2,t})=(1/2,1/2)$ if $s_{1,t-1}=s_{2,t-1}=1/2$;
$(\alpha_{1,t},\alpha_{2,t})=(s_{1,t-1},1-s_{1,t-1})$ if player $1$ was committed and player $2$ was flexible; 
$(\alpha_{1,t},\alpha_{2,t})=(1-s_{2,t-1},s_{2,t-1})$ if player $1$ was flexible and player $2$ was committed; 
and $(\alpha_{1,t},\alpha_{2,t})=(\alpha_{1,t-1},\alpha_{2,t-1})$ if both players were flexible.

Each player $i$ maximizes their discounted average payoff $\sum_{t=0}^{+\infty} (1-\delta)\delta^t v_{i,t}$, where $v_{i,t}$ stands for player $i$'s stage-game payoff in period $t$.

Players can observe the current status quo (i.e., the \textit{state}) before choosing their actions.\footnote{Even when players can observe their opponents' actions in previous periods, they have no incentive to condition their behavior on such information given that their opponent does not do that, which is indeed the case in all MPEs.}  A \textit{Markov strategy} for player $i$ is $\sigma_i : [0,1] \rightarrow \Delta(S)$ where $\sigma_i (\alpha_i) \in \Delta(S)$ is his stage-game action when his share under the current status quo  is $\alpha_i$. A \textit{Markov Perfect equilibrium} (or MPE) is a Markov strategy profile $(\sigma_1,\sigma_2)$ such that for every $i \in \{1,2\}$, $\sigma_i$ best replies to $\sigma_{-i}$ starting from any status quo $(\alpha_i,\alpha_{-i})$. An MPE is \textit{symmetric} if both players use the same mapping, i.e., $\sigma_1=\sigma_2$.

\paragraph{Benchmark:} Consider a benchmark setting where the status quo (the division of surplus when both players are flexible) evolves according to an exogenous process. Proposition \ref{Prop1} shows that in every MPE, players will never commit to any demand other than $1$ and that the payoffs players receive in the resulting equilibria are bounded away from the socially efficient payoff. 
\begin{Proposition}\label{Prop1}
If the evolution of state variable $(\alpha_{1,t},\alpha_{2,t})_{t \in \mathbb{N}}$ is independent of players' actions, then in every MPE, no player will commit to any demand that is strictly less than $1$, and 
the sum of players' discounted average payoffs is no more than $1$ in any MPE under any $c>0$. 
\end{Proposition}
The intuition is that when the state evolution is independent of players' actions, each player's expected continuation value in the next period is independent of his action in the current period. Hence, each player's incentive in the dyanmic game coincides with that in the stage game. Proposition \ref{Prop1} then follows from the conclusion in Ellingsen and Miettinen (2008), that when player $i$'s status quo is $\alpha_i$, committing to anything weakly less than $\alpha_i$ is strictly dominated by being flexible, and after those actions are eliminated, committing to any $\beta_i <1$ is strictly dominated by the mixed action of committing to $1$ with probability $\beta_i$ and being flexible with probability $1-\beta_i$.

Once all strategies  except for $f$ and committing to $1$ are eliminated, there are two types of MPEs (i) one player is flexible and another player commits to $1$, in which case the sum of their payoffs is $1-c$, and (ii)
 both players mix between $f$ and committing to $1$ with probabilities that make the other player indifferent. Such mixed-strategy equilibria exist if and only if the status quo satisfies $1-c > \max\{u(\alpha_1),u(\alpha_2)\}$. The sum of players' equilibrium payoffs is at most
 \begin{equation}\label{2.1}
 c \cdot \sup_{\alpha \in A} \Big\{\frac{u(\alpha)}{1-u(\alpha)},\frac{u(1-\alpha)}{1-u(1-\alpha)} \Big\},
\end{equation}
where $A \equiv \{\alpha \in [0,1]  | \max\{ u(\alpha),u(1-\alpha) \}<1-c\}$.
The maximum in (\ref{2.1}) is attained when the constraint $\max\{ u(\alpha),u(1-\alpha) \} \leq 1-c$ binds and this maximum is less than $1$.

\section{Analysis}\label{sec3}
Our main result focuses on situations where players are patient and their costs of commitment are small. Fix players' utility function $u$ and let
$(\sigma_1(\delta,c),\sigma_{2}(\delta,c))$ denote a generic MPE under parameter configuration $(\delta,c)$. 
A \textit{class} of MPEs $\{\sigma_1(\delta,c),\sigma_2(\delta,c)\}_{\delta \in (0,1),c >0}$ consists of one MPE for each $(\delta,c) \in (0,1) \times (0,+\infty)$.

We say that a class of MPEs $\{\sigma_1(\delta,c),\sigma_2(\delta,c)\}_{\delta \in (0,1),c >0}$ attain \textit{asymptotic efficiency} if for every $\varepsilon>0$, there exist $\delta^* \in (0,1)$ and $c^* >0$ such that when $\delta \in (\delta^*,1)$ and $c<c^*$, the sum of the two players' discounted average payoffs in those equilibria is greater than $2u(1/2)-\varepsilon$. When $u$ is strictly concave, asymptotic efficiency is achieved if and only if the division is fair in almost all periods (in terms of the occupation measure).

We say that a class of MPEs $\{\sigma_1(\delta,c),\sigma_2(\delta,c)\}_{\delta \in (0,1),c >0}$ are \textit{asymptotic renegotiation proof} if for every $\varepsilon>0$, there exist $\delta^* \in (0,1)$ and $c^* >0$ such that when $\delta \in (\delta^*,1)$ and $c<c^*$, for every status quo $(\alpha_1,\alpha_2)$, players' equilibrium (expected) stage-game payoffs at $(\alpha_1,\alpha_2)$ and their continuation values at $(\alpha_1,\alpha_2)$
are both $\varepsilon$-close to the Pareto payoff frontier $\mathcal{U}$.\footnote{Our notion of asymptotic renegotiation proofness is neither stronger nor weaker than the two notions of renegotiation proofness introduced by Farrell and Maskin (1989): Their strong renegotiation proofness requires none of the continuation equilibria to be strictly Pareto dominated by another renegotiation proof equilibrium, whereas we require both the stage-game payoffs and the continuation values belonging to the \textit{Pareto frontier}. However, we only require renegotiation proofness \textit{in the limit}. Instead, they require exact renegotiation proofness under a fixed discount factor.}

We view asymptotic renegotiation proofness as a reasonable refinement in bargaining games. It requires that \textit{regardless} of the current-period status quo, including those that are never reached on the equilibrium path, the amount of surplus that is lost in the stage game and in the continuation game are both negligible. 
Our notion of renegotiation proofness is built on two hypothesis:
First, players cannot credibly commit to burn a significant fraction of the surplus in the current period for the purpose of providing incentives.\footnote{Under our notion of renegotiation proofness, players can renegotiate in the beginning of each period \textit{before} they make commitments. Once they commit to a share of the surplus (e.g., taking the form of burning bridges), it is too costly (if not impossible) for them to back down and to renegotiate with their opponent for better allocations.} Second, players will coordinate on better outcomes in the future once there are significant gains from doing so.

It is also worth noticing that the requirement that players' expected stage-game payoffs being close to $\mathcal{U}$ and the requirement that their expected continuation values being close to $\mathcal{U}$ do not imply one another. This is because first, when players' stage-game payoff profile alternates between $(1,0)$ and $(0,1)$, their stage-game payoffs are Pareto efficient but their continuation values are bounded below $\mathcal{U}$. Second, even when players' continuation values are close to the Pareto frontier $\mathcal{U}$, their stage-game payoffs might be bounded away $\mathcal{U}$ provided that their stage-game payoffs have negligible impact on their continuation values once $\delta$ is close to $1$.

Our renegotiation proofness requirement together with the Markovian restriction makes it harder to achieve efficiency. The reason is that these restrictions limit the scope of providing incentives via coordinating on low-payoff continuation equilibria at off-path histories, which is sometimes used when showing folk theorems and efficiency results in repeated and stochastic games (e.g.,  Fudenberg and Maskin 1986, Abreu 1988, Escobar and Toikka 2013).

Despite all these challenges, our main result, Theorem \ref{Theorem1}, shows that asymptotic efficiency can still be attained in Markov Perfect equilibria that are asymptotic renegotiation proof. 
\begin{Theorem}\label{Theorem1}
For any $\alpha_0 \in (0,1)$ and any strictly concave $u$, there exist a class of asymptotic renegotiation proof MPEs that achieve asymptotic efficiency.
\end{Theorem}
The proof is in Section \ref{sub3.1}. Our proof is constructive, which sheds light on how the status quo effect motivates players to commit to moderate demands and why it is plausible that players can commit to moderate demands with probability close to $1$ in equilibrium.

Theorem \ref{Theorem1} implies that there exist a class of equilibria in which as players become patient and their costs of making commitments vanish, the surplus is fairly divided between the players in almost all periods. This conclusion stands in contrast to that in static bargaining games and dynamic bargaining games with an exogenously evolving status quo, in which Proposition \ref{Prop1} shows that players will never commit to moderate demands in any MPE and as a result, the sum of their equilibrium payoffs is bounded below efficiency even in the limit.

Our conclusion also stands in contrast to the existing results on dynamic legislative bargaining, which also examine dynamic games where the current division 
of surplus affects the future status quo. Some of those papers such as Bowen and Zahran (2012) and Richter (2014) show that players will make moderate demands in equilibria where there are socially inefficient punishments under some values of the status quo. In contrast, due to our asymptotic renegotiation proofness requirement, inefficient punishments are not allowed both on and off the equilibrium path.

In our game where the status quo is determined by precedents, it is not obvious how one can achieve social efficiency without the need to coordinate on Pareto inefficient punishments at some off-path status quo. For example, suppose the current status quo is $(1/2,1/2)$, the socially efficient division can be achieved when no player commits to demands that are strictly greater than $1/2$. 
However, if one player commits to $1/2$, then it provides his opponent an incentive to remain flexible since doing so can save his cost of making commitments. But when a player is flexible, his opponent will have an incentive to demand the entire surplus $1$ instead of $1/2$, as long as a player's continuation value is higher when his status quo allocation is $1$ compared to $1/2$. If players cannot commit to destroy surplus when the status quo is $(1,0)$, then it seems to be the case that a player will strictly prefer to commit to $1$ rather than $1/2$ given that his opponent is flexible.

In addition, our result is not implied by the folk theorem results and the efficiency results in repeated and stochastic games, such as the ones in 
Fudenberg and Maskin (1986), Abreu (1988),  and H\"{o}rner, Sugaya, Takahashi, and 
Vieille (2011). This is because we restrict attention to MPEs, which forbids arbitrary history-dependent punishments. We also require the equilibria to be asymptotic renegotiation proof, which rule out Pareto inefficient punishments. At a technical level, H\"{o}rner, Sugaya, Takahashi, and 
Vieille (2011)'s folk theorem in stochastic games requires that the number of states being finite and that 
the state evolving according to an irreducible Markov chain for all
strategy profiles. Both conditions are violated in our model.

\subsection{Proof of Theorem 1}\label{sub3.1}
We construct a class of symmetric asymptotic renegotiation proof MPEs that achieve asymptotic efficiency when the status quo in period $0$ is $(1/2,1/2)$. In Appendix \ref{secA}, we generalize our construction to any initial status quo where both players' shares are strictly positive.

In each of the equilibria we construct and at every status quo $(\alpha_1,\alpha_2)$, each player will take at most the following three actions with positive probability: remaining flexible, committing to $1$, and committing to $1/2$. Due to symmetry, we use $V(\alpha)$ to denote a player's continuation value when his share under the current status quo is $\alpha$. We start from the following lemma:
\begin{Lemma}\label{L1}
In any symmetric MPE and for any $\alpha \in [0,1]$, when the status quo is $(\alpha,1-\alpha)$: 
\begin{enumerate}
\item A player will commit to $\beta \in (1/2,1]$ with zero probability at  $(\alpha,1-\alpha)$ unless $\beta$ satisfies
\begin{equation}\label{3.1}
\beta \in \arg\max_{\beta' \in (1/2,1]} \{(1-\delta)u(\beta') + \delta V(\beta')\}.
\end{equation}
\item A player will commit to $1/2$ with zero probability at  $(\alpha,1-\alpha)$ unless 
\begin{equation}\label{3.2}
1/2 \in \arg\max_{\beta' \in [1/2,1]} \{(1-\delta)u(\beta') + \delta V(\beta')\},
\end{equation} or his opponent commits to $1/2$ with positive probability at 
$(\alpha,1-\alpha)$. 
\end{enumerate}
\end{Lemma}
\begin{proof}
Fix any $\beta \notin \arg\max_{\beta' \in (1/2,1]} \{(1-\delta)u(\beta') + \delta V(\beta')\}$ and $\alpha \in [0,1]$.
Suppose by way of contradiction that player $1$ commits to demand $\beta \in (1/2,1]$ with positive probability when the status quo is $(\alpha,1-\alpha)$. Player $1$'s discounted average payoff is $(1-\delta) (u(\beta)-c) + \delta V(\beta)$ if player $2$ chooses $f$, and is $-(1-\delta)c+ \delta V(\alpha)$ if player $2$ commits. Since $\beta \notin \arg\max_{\beta' \in (1/2,1]} \{(1-\delta)u(\beta') + \delta V(\beta')\}$, there exists $\beta^* \in (1/2,1]$ that leads to a strictly higher payoff than $\beta$ when player $2$ chooses $f$ and leads to the same payoff as $\beta$ when player $2$ commits. This implies that $\beta$ is not optimal for player $1$ at $(\alpha,1-\alpha)$ unless player $2$ plays $f$ at $(\alpha,1-\alpha)$ with zero probability. 

Suppose by way of contradiction that at $(\alpha,1-\alpha)$, player $1$ commits to $\beta$ with positive probability and
player $2$ plays $f$ with zero probability. If this is the case, then it is optimal for player $1$ to commit to $\beta$ at $(\alpha,1-\alpha)$, so his 
 continuation value at $(\alpha,1-\alpha)$ satisfies
$V(\alpha)=-(1-\delta)c + \delta V(\alpha)$, 
which implies that $V(\alpha)=-c$. This leads to a contradiction since player $1$ can secure payoff $0$ by playing $f$ in every period. 

Similarly, suppose by way of contradiction that there exist an MPE and a status quo $(\alpha,1-\alpha)$ at which player $1$ commits to $1/2$ with strictly positive probability despite player $2$ commits to $1/2$ with zero probability and (\ref{3.2}) fails. Then there exists $\beta^* >1/2$ such that committing to $\beta^*$ leads to a strictly higher payoff for player $1$ than committing to $1/2$ when player $2$ chooses $f$ and leads to the same payoff for player $1$ as committing to $1/2$ when player $2$ commits to anything strictly greater than $1/2$. If at $(\alpha,1-\alpha)$, player $1$ commits to $1/2$ with positive probability and
player $2$ plays $f$ and commits to $1/2$ with zero probability, player $1$'s continuation value at $(\alpha,1-\alpha)$ satisfies
$V(\alpha)=-(1-\delta)c + \delta V(\alpha)$, 
which implies that $V(\alpha)=-c$. This leads to a contradiction.
\end{proof}
In what follows, we will focus on an \textit{auxiliary game} in which each player can only choose between $f$, committing to $1/2$, and committing to $1$. Later on, we will verify that in the equilibrium we construct, 
\begin{equation}\label{3.3}
\arg\max_{\beta' \in (1/2,1]} \{(1-\delta)\beta' + \delta V(\beta')\}=\{1\},
\end{equation}
which implies that no player will commit to any $\beta \in (1/2,1)$ even when they can do so. 
\begin{Lemma}\label{L2}
For every $\varepsilon>0$, there exist $\delta^* \in (0,1)$ and $c^*>0$ such that when $\delta< \delta^*$ and $c<c^*$, there exists a symmetric MPE in the auxiliary game such that:
\begin{enumerate}
\item When the status quo is $(1,0)$, both players mix between $1$ and $f$. 
\item When the status quo is $(1/2,1/2)$, both players mix between $1$, $1/2$, and $f$, and each player's probability of playing $1/2$ is more than $1-\varepsilon$. 
\item Continuation values satisfy $V(0)=0$, $V(1/2) \in (u(1/2)-\varepsilon, u(1/2))$, and $V(1) \in (1-\varepsilon,1)$. 
\end{enumerate}
\end{Lemma}
\begin{proof}
First, let us consider histories where the current-period status quo is $(1,0)$. When players can only choose between $1$ and $f$, 
their discounted average payoffs are given by
\begin{center}
\begin{tabular}{| c | c | c | }
  \hline
  -- & $f$ & $1$\\
  \hline
  $f$ & $(1-\delta)+\delta V(1),\delta V(0)$  & $\delta V(0), (1-\delta)(1-c)+ \delta V(1)$ \\
  \hline
  $1$ & $(1-\delta)(1-c)+\delta V(1), \delta V(0)$ & $-(1-\delta)c+ \delta V(1), -(1-\delta)c+\delta V(0)$\\
  \hline
\end{tabular}
\end{center}
We solve for a mixed-strategy equilibrium in this auxiliary game where both players mix between $f$ and $1$ with positive probability. 
Later on, we will verify that no player has any incentive to play $1/2$ when the status quo is $(1,0)$ or $(0,1)$. 
In any such equilibrium, playing $f$ is optimal for player $2$, which implies that $V(0)=\delta V(0)$, and therefore, $V(0)=0$. Let $p$ denote the probability that player $1$ plays $f$ and let $q$ denote the probability that player $2$ plays $f$. 
Player $1$'s indifference condition and promise-keeping condition lead to the following system of equations:
\begin{equation}\label{3.4}
V(1)=(1-\delta)(q-c) + \delta V(1)= q \Big\{ (1-\delta) + \delta V(1) \Big\}.  
\end{equation}
This implies that $V(1)=q-c$. Plugging this back to (\ref{3.4}), we have
\begin{equation}\label{3.5}
(1-q\delta) (q-c)=q(1-\delta) \quad \Leftrightarrow \quad q^2 -(c+1) q + \frac{c}{\delta}=0.
\end{equation}
Solving the quadratic equation (\ref{3.5}) and applying the Taylor's expansion formula at $\delta=1$, we obtain the following two solutions for $q$:
\begin{equation*}
q_- = c+ \frac{1-\delta}{\delta} \cdot \frac{c}{1-c}+ \mathcal{O}((1-\delta)^2) \textrm{ and }
q_+ = 1-\frac{1-\delta}{\delta} \cdot \frac{c}{1-c}- \mathcal{O}((1-\delta)^2).
\end{equation*}
For the larger solution $q_+$, player $1$'s continuation value is given by
\begin{equation}\label{3.6}
V(1)=1-c-\frac{1-\delta}{\delta} \cdot \frac{c}{1-c}- \mathcal{O}((1-\delta)^2).
\end{equation}
Player $2$'s indifference condition requires that
\begin{equation*}
(1-\delta)(p-c)+ \delta p V(1)=0.
\end{equation*}
There exists $p \in (0,1)$ when 
$c$ is close to $0$, $\delta$ is close to $1$, and
$V(1)$ is given by (\ref{3.6}). 

Next, let us consider histories where the status quo is $(1/2,1/2)$.  We solve for a symmetric mixed-strategy equilibrium in this auxiliary game where both players mix between $f$, $1/2$, and $1$ with positive probability while taking into account that their continuation values satisfy $V(0)=0$ and $V(1)=1-c-\frac{1-\delta}{\delta} \cdot \frac{c}{1-c}- \mathcal{O}((1-\delta)^2)$. Let $r$ denote the probability that each player plays $f$ and $s$ denote the probability that each player plays $1/2$. Each player's indifference conditions and promise keeping condition imply that:
\begin{eqnarray}\label{3.7}
V(1/2) & = & (r+s) \Big( (1-\delta)u(1/2)+ \delta V(1/2) \Big)
\nonumber\\
& = & {}  (1-\delta) \Big( (r+s) u(1/2)-c \Big) + \delta V(1/2)
\nonumber\\
& = & {} (1-\delta) (r-c)+ \delta \Big( rV(1)+(1-r)V(1/2) \Big).
\end{eqnarray}
This implies that $V(1/2)=(r+s) u(1/2)-c$. Replacing $V(1/2)$ with $(r+s)u(1/2)-c$ in $V(1/2)=(1-\delta) \Big( (r+s) u(1/2)-c \Big) + \delta V(1/2)$, we have
\begin{equation*}
u(1/2) (r+s)^2 - (u(1/2)+c) (r+s) + \frac{c}{\delta}=0.
\end{equation*}
Solving this quadratic equation of $r+s$ and applying the Taylor's expansion formula at $\delta=1$, 
\begin{equation*}
(r+s)_-= \frac{c}{u(1/2)} + \frac{1-\delta}{\delta} \cdot \frac{c}{u(1/2)-c}+ \mathcal{O}((1-\delta)^2) \textrm{ and }
(r+s)_+= 1-\frac{1-\delta}{\delta} \cdot \frac{c}{u(1/2)-c}-\mathcal{O}((1-\delta)^2).
\end{equation*}
Take the larger solution for $r+s$ and plug into $V(1/2)=(r+s) u(1/2)-c$, we have
\begin{equation}\label{3.8}
V(1/2)= \Big( 1-\frac{1-\delta}{\delta} \cdot \frac{c}{u(1/2)-c} \Big) u(1/2)-c-\mathcal{O}((1-\delta)^2).
\end{equation}
The promise-keeping condition $V(1/2)=(1-\delta) (r-c)+ \delta \Big( rV(1)+(1-r)V(1/2) \Big)$ implies that
\begin{equation*}
(1-\delta) (V(1/2)+c) = r\Big\{
(1-\delta) + \delta \big(V(1)-V(1/2) \big)
\Big\},
\end{equation*}
which implies that when $c$ is small enough and $\delta$ is close to $1$, $r= \mathcal{O}(1-\delta)$. This together with the fact that the larger solution for $r+s$ converging to $1$ as $c \rightarrow 0$ and $\delta \rightarrow 1$ implies that $s$ converges to $1$ in the limit. 

Next, we verify that when $c$ is small enough and $\delta$ is close to $1$, at status quo $(1,0)$, 
no player has any incentive to play $1/2$ when his opponent plays $1/2$ with zero probability and players' continuation values satisfy $V(0)=0$, (\ref{3.6}), and (\ref{3.8}). This is because compared to playing $1$, playing $1/2$ leads to a strictly lower payoff when the opponent plays $f$, and leads to a weakly lower payoff when the opponent plays $1$. Therefore, playing $1/2$ is strictly suboptimal unless the opponent plays $1$ with probability $1$. 
Suppose by way of contradiction that there exists an MPE such that when the status quo is $(1,0)$, one of the players $i \in \{1,2\}$ plays $1/2$ with positive probability and his opponent plays $1$ for sure. Then this player's continuation value $V$ satisfies $V=-(1-\delta)c+\delta V$, which implies that $V=-c$. This leads to a contradiction. 
\end{proof}
We construct a class of asymptotic renegotiation proof MPEs in the original game such that  players' behaviors when the status quo is $(1/2,1/2)$, $(0,1)$, and $(1,0)$ are given by the MPE we constructed in the auxiliary game. Those equilibria are asymptotically efficient since both players' continuation values when the status quo is $(1/2,1/2)$ are close to $u(1/2)$ as $\delta \rightarrow 1$ and $c \rightarrow 0$. 

Recall the value of $V(1)$ in Lemma \ref{L2}.
Let $\alpha^* \in (1/2,1)$ be defined via 
\begin{equation}\label{3.9}
u(\alpha^*)=(1-\delta)(1-c)+ \delta V(1). 
\end{equation}
Such $\alpha^*$ exists when $c \rightarrow 0$ and $\delta \rightarrow 1$ since $u(1/2)<(1-\delta)(1-c)+ \delta V(1)<u(1)=1$. 

Consider an auxiliary game where the initial status quo is $(\alpha,1-\alpha)$ with $\alpha \in (1/2,1)$ and that players can only choose between $f$ and committing to $1$ and remaining flexible. Their payoffs are
\begin{center}
\begin{tabular}{| c | c | c | }
  \hline
  -- & $f$ & $1$\\
  \hline
  $f$ & $(1-\delta)u(\alpha)+\delta V(\alpha),(1-\delta) u(1-\alpha) + \delta V(1-\alpha)$  & $\delta V(0), (1-\delta)(1-c)+ \delta V(1)$ \\
  \hline
  $1$ & $(1-\delta)(1-c)+\delta V(1), \delta V(0)$ & $-(1-\delta)c+ \delta V(1), -(1-\delta)c+\delta V(0)$\\
  \hline
\end{tabular}
\end{center}

For every $\alpha \in (1/2,\alpha^*]$, when the status quo is $(\alpha,1-\alpha)$, player $1$ commits to $1$ and player $2$ plays $f$. It is easy to verify that this is an equilibrium in the auxiliary game. 
Under such an equilibrium behavior, we know that for every $\alpha \in (1/2,\alpha^*]$, players' continuation values satisfy $V(1-\alpha)=0$ and 
$V(\alpha)=(1-\delta)(1-c)+ \delta V(1)$. Therefore, 
\begin{equation*}
(1-\delta)u(\alpha)+ \delta V(\alpha) \leq \underbrace{ (1-\delta)u(\alpha^*)+ \delta V(\alpha^*)=(1-\delta) (1-c)+ \delta V(1) }_{\textrm{from (\ref{3.9})}}< (1-\delta) + \delta V(1),
\end{equation*}
which implies that $(1-\delta)u(\alpha)+ \delta V(\alpha) < (1-\delta)  +\delta V(1)$.

For every $\alpha \in (\alpha^*,1)$, in the auxiliary game where the starting status quo is $(\alpha,1-\alpha)$, there exists an equilibrium in which both players play both $f$ and $1$ with positive probability at $(\alpha,1-\alpha)$. This is because (i) player $1$ strictly prefers $f$ when player $2$ plays $f$ and strictly prefers to commit to $1$ when player $2$ commits to $1$, and (ii) player $2$ strictly prefers $f$ when player $1$ commits to $1$ and strictly prefers to commit to $1$ when player $1$ plays $f$. Lemma \ref{L3} provides an upper bound and a lower bound on player $1$'s continuation value at $(\alpha,1-\alpha)$ in such mixed-strategy equilibria. 
\begin{Lemma}\label{L3}
There exist $\delta^* \in (0,1)$ and $c^*>0$ such that when $\delta > \delta^*$ and $c<c^*$, 
there exists a mixed-strategy equilibrium in the auxiliary game in which for any $\alpha \in (\alpha^*,1)$, 
player $1$'s continuation value $V(\alpha)$ satisfies $V(\alpha)>1-\varepsilon$ and 
\begin{equation}\label{3.10}
(1-\delta) u(\alpha) + \delta V(\alpha) \leq (1-\delta) + \delta V(1),
\end{equation}
and that the probability with which player $2$ plays $f$ at $(\alpha,1-\alpha)$ is more than $1-\varepsilon$. 
\end{Lemma}
After showing Lemma \ref{L3}, let us consider the original game where players can commit to anything in $[1/2,1]$ and suppose under every status quo $(\alpha,1-\alpha)$, players behave according to the behaviors we described in the auxiliary game. Lemma \ref{L3} together with our earlier conclusion that $(1-\delta)u(\alpha)+ \delta V(\alpha) < (1-\delta)  +\delta V(1)$ for every $\alpha \in [1/2,1)$ implies that $1$ is the unique maximizer for $(1-\delta) u(\alpha)+ \delta V(\alpha)$. Lemma \ref{L1}
implies that players have no incentive to commit to anything strictly between $1/2$ and $1$, and will have no incentive to commit to $1/2$ unless the status quo is $(1/2,1/2)$ since the other player commits to $1/2$ with zero probability under those status quo. The class of MPEs we constructed is asymptotically renegotiation proof since $V(\alpha) \rightarrow 1$ for all $\alpha >1/2$ and $V(\alpha) \rightarrow u(1/2)$ when $\alpha=1/2$. The rest of the proof shows Lemma \ref{L3}.
\begin{proof}[Proof of Lemma 3:]
First, in every such mixed-strategy equilibrium, player $1$'s continuation value at $\alpha$ is at least his payoff from committing to $1$, which gives:
\begin{equation*}
V(\alpha) \geq \min \{(1-\delta)(1-c)+\delta V(1) , -(1-\delta)c+ \delta V(1)\} \geq -(1-\delta) c+ \delta V(1).
\end{equation*}
The expression for $V(1)$ in (\ref{3.6}) implies that $V(\alpha)$ converges to $1$ as $\delta \rightarrow 1$ and $c \rightarrow 0$.

Next, suppose by way of contradiction that there exists
$\alpha \in (\alpha^*,1)$ such that 
(\ref{3.10}) fails. This hypothesis implies that
\begin{equation}\label{3.11}
\Big[
(1-\delta) u(\alpha) + \delta V(\alpha)
\Big] -\Big[
(1-\delta)(1-c)+ \delta V(1)
\Big] \geq (1-\delta)c.
\end{equation}
Let $p$ denote the probability that player $2$ plays $f$ when the status quo is $(\alpha,1-\alpha)$. 
Since when player $2$ plays $1$, player $1$'s payoff is $0$ when he plays $f$ and is $-(1-\delta) c + \delta V(\alpha)$ if he commits to $1$, 
player $1$ being indifferent between $f$ and committing to $1$ implies that
\begin{equation*}
p= \frac{-(1-\delta)c + \delta V(\alpha)}{-(1-\delta)c + \delta V(\alpha) + [
(1-\delta) u(\alpha) + \delta V(\alpha)
] -[
(1-\delta)(1-c)+ \delta V(1)]}. 
\end{equation*}
By (\ref{3.11}), we have
\begin{equation*}
p \leq \frac{-(1-\delta)c + \delta V(\alpha)}{\delta V(\alpha)}=1-\frac{1-\delta}{\delta} \cdot \frac{c}{V(\alpha)}.
\end{equation*}
Player $1$'s promise-keeping condition and the fact that $f$ being optimal require that
\begin{eqnarray*}
V(\alpha) &=& p \Big[ (1-\delta) u(\alpha) + \delta V(\alpha) \Big] \leq \Big(1-\frac{1-\delta}{\delta} \cdot \frac{c}{V(\alpha)}\Big)\Big[ (1-\delta) u(\alpha) + \delta V(\alpha) \Big]
\nonumber\\
& = & {} (1-\delta) u(\alpha) +\delta V(\alpha) -(1-\delta)c -\frac{(1-\delta)^2}{\delta} \cdot \frac{cu(\alpha)}{V(\alpha)}.
\end{eqnarray*}
Simplifying the above inequality, we obtain that
\begin{equation}\label{3.12}
V(\alpha) \leq u(\alpha)-c- \frac{1-\delta}{\delta} \cdot \frac{c u(\alpha)}{V(\alpha)}. 
\end{equation}
Our hypothesis that (\ref{3.10}) fails when the status quo is $(\alpha,1-\alpha)$ implies that
\begin{equation}\label{3.13}
V(\alpha) \geq V(1)+ \frac{1-\delta}{\delta} (1-u(\alpha)) = 1-c+ \frac{1-\delta}{\delta} \Big(1-u(\alpha)-\frac{c}{1-c} \Big). 
\end{equation}
Inequalities (\ref{3.12}) and (\ref{3.13}) together imply that
\begin{equation}\label{3.14}
1-u(\alpha) \leq (1-\delta) c \Big(\frac{1}{1-c}- \frac{u(\alpha)}{V(\alpha)}\Big).
\end{equation}
If player $1$ commits to $1$ with positive probability when the status quo is $(\alpha,1-\alpha)$, then his promise-keeping condition requires that
\begin{equation*}
V(\alpha)= (1-\delta) (p-c) + \delta (pV(1)+(1-p)V(\alpha)) \Leftrightarrow V(\alpha)=\frac{(1-\delta) (p-c) +\delta p V(1)}{1-\delta + \delta p}.
\end{equation*}
This implies that $V(\alpha) < 1-c$ given that $u(\alpha)-c<1-c$ and our earlier conclusion that $V(1)<1-c$. Since the RHS of (\ref{3.14}) is strictly increasing in $V(\alpha)$,  (\ref{3.14}) implies that 
\begin{equation*}
1-u(\alpha) \leq (1-\delta) c \Big(\frac{1}{1-c}- \frac{u(\alpha)}{V(\alpha)}\Big) \leq (1-\delta) c \frac{1-u(\alpha)}{1-c} \Leftrightarrow 1 \leq \frac{c(1-\delta)}{1-c}.
\end{equation*}
There exists $\delta^* \in (0,1)$ and $c^*>0$ such that when $\delta > \delta^*$ and $c<c^*$, 
the above inequality is not true. This contradiction implies that (\ref{3.10}) holds for all $\alpha \in (\alpha^*,1)$. 

Since (\ref{3.10}) holds for all $\alpha \in (\alpha^*,1)$, we have
\begin{equation*}
\Big[
(1-\delta) u(\alpha) + \delta V(\alpha)
\Big] -\Big[
(1-\delta)(1-c)+ \delta V(1)
\Big] < (1-\delta)c,
\end{equation*}
which implies the following lower bound on $p$:
\begin{equation*}
p \geq \frac{-(1-\delta)c + \delta V(\alpha)}{\delta V(\alpha)}=1-\frac{1-\delta}{\delta} \cdot \frac{c}{V(\alpha)} \geq 
1-\frac{1-\delta}{\delta} \cdot \frac{c}{-(1-\delta) c+ \delta V(1)}.
\end{equation*}
The RHS converges to $1$ as $c \rightarrow 0$ and $\delta \rightarrow 1$, which implies that player $2$ plays $f$ with probability close to $1$ in the limit. 
\end{proof}
\newpage
\appendix
\section{Proof of Theorem 1: General Values of the Initial Status Quo}\label{secA}
We extend our construction in Section \ref{sub3.1} by allowing for any initial status quo $(\alpha_0,1-\alpha_0)$ with $\alpha_0 \in (0,1)$. Without loss of generality, we assume that $\alpha_0  \in (1/2,1)$. When the current status quo $(\alpha,1-\alpha)$ satisfies $\alpha \in [0,1-\alpha_0) \cup \{1/2\} \cup (\alpha_0,1]$, players' strategies remain the same as before, which implies that their continuation value $V(\alpha)$ remains the same at those values of $\alpha$. When $\alpha \in [1-\alpha_0,1/2) \cup (1/2,\alpha_0]$, both players randomize between $f$, committing to $1/2$, and committing to $1$, with probabilities that make each other indifferent. Players' continuation values when  $\alpha \in [1-\alpha_0,1/2) \cup (1/2,\alpha_0]$ converge to $u(1/2)$ as $\delta \rightarrow 1$ and $c \rightarrow 0$. One can also verify that $(1-\delta)\alpha + \delta V(\alpha)< (1-\delta) + \delta V(1)$ for every such $\alpha$, which by Lemma \ref{L1}, implies that players have no incentive to commit to anything other than $1/2$ and $1$ at any status quo. 

In what follows, we show that there exists an equilibrium in which players' continuation value at any $\alpha \in [1-\alpha_0,1/2) \cup (1/2,\alpha_0]$ is close to $u(1/2)$ as $\delta \rightarrow 1$ and $c \rightarrow 0$. Fix any such $\alpha$, let $r$ denote the probability that player $2$ plays $f$ at $(\alpha,1-\alpha)$ and let $s$ denote the probability that player $2$ commits to $1/2$ at $(\alpha,1-\alpha)$. With complementary probability, player $2$ commits to $1$. Player $1$ is indifferent between $f$, committing to $1/2$, and committing to $1$ at $(\alpha,1-\alpha)$ if and only if
\begin{eqnarray}\label{A.1}
V(\alpha) & = & r \Big\{ (1-\delta) u(\alpha) +\delta V(\alpha)  \Big\} + s \Big\{ (1-\delta) u(1/2) + \delta V(1/2)  \Big\}
\nonumber\\
& = & {}  (1-\delta) \Big\{ (r+s) u(1/2)-c \Big\} + \delta \Big\{ (r+s) V(1/2) + (1-r-s) V(\alpha) \Big\}
\nonumber\\
& = & {} (1-\delta) (r-c)+ \delta \Big\{ rV(1)+(1-r)V(\alpha) \Big\}.
\end{eqnarray}
Let $W(1/2) \equiv (1-\delta) u(1/2) + \delta V(1/2)$ and $W(1) \equiv (1-\delta) + \delta V(1)$, which are constants derived in Section \ref{sub3.1}, 
(\ref{A.1}) leads to the following system of equations for $(r,s,V(\alpha))$:
\begin{equation}\label{A.2}
r+s = \frac{(1-\delta) (V(\alpha)+c)}{W(1/2)-\delta V(\alpha)},
\end{equation}
\begin{equation}\label{A.3}
r= \frac{(1-\delta) (V(\alpha)+c)}{W(1)-\delta V(\alpha)},
\end{equation}
\begin{equation}\label{A.4}
V(\alpha)= r \Big\{ (1-\delta) u(\alpha) +\delta V(\alpha)  \Big\} + s W(1/2). 
\end{equation}
The RHS of both (\ref{A.2}) and (\ref{A.3}) are strictly increasing in $V(\alpha)$. 
Using our earlier conclusion that $(1-\delta) + \delta V(1)= 1-c-\mathcal{O}((1-\delta)c^2)$ and 
$(1-\delta)u(1/2) + \delta V(1/2)= u(1/2)-c-\mathcal{O}((1-\delta)c^2)$, we know that  
if $V(\alpha)=W(1/2)$ and $r$ and $s$ are derived via (\ref{A.2}) and (\ref{A.3}), then the RHS of (\ref{A.4}) is strictly greater than the LHS as $\delta \rightarrow 1$ and $c \rightarrow 0$. Moreover, for any $\varepsilon>0$, there exist $\delta^*$ and $c^*$ such that when $\delta> \delta^*$ and $c<c^*$, if $V(\alpha)<W(1/2)-\varepsilon$ and $r$ and $s$ are derived via (\ref{A.2}) and (\ref{A.3}), then the RHS of (\ref{A.4}) is strictly less than the LHS as $\delta \rightarrow 1$ and $c \rightarrow 0$. The intermediate value theorem implies that there exists a solution where $V(\alpha)$ converges to $u(1/2)$ as 
$\delta \rightarrow 1$ and $c \rightarrow 0$.

\end{spacing}

\end{document}